\documentclass[letterpaper, 10 pt, conference]{ieeeconf}  

\IEEEoverridecommandlockouts                              

\usepackage{amsmath, amsthm, mathtools} 
\usepackage{amssymb}  
\usepackage{amsfonts, dsfont}
\usepackage{color, xcolor}
\usepackage{graphicx}
\usepackage{nicefrac}
\allowdisplaybreaks
\usepackage{url}
\newcommand{\A}{\mathcal{A}}
\newcommand{\B}{\mathcal{B}}

\newcommand{\M}{\mathcal{M}}
\renewcommand{\P}{\mathcal{P}}

\newcommand{\R}{\mathcal{R}}
\renewcommand{\S}{\mathcal{S}}
\newcommand{\T}{\mathcal{T}}

\newcommand{\X}{\mathcal{X}}

\newcommand{\pr}[1]{\mathbb{P}\left\{ #1 \right\}}
\newcommand{\KL}{\mathrm{KL}}
\newcommand{\dTV}{\mathrm{d}_{\mathrm{TV}}}
\newcommand{\dM}{\mathrm{d}_{\mathcal{M}}}
\newcommand{\dPi}{\mathrm{d}_\Pi}

\newcommand{\opt}{\mathrm{opt}}
\newcommand{\prop}{\mathrm{prop}}

\newcommand{\abs}[1]{\halfspace \vert  #1 \vert}            
\newcommand{\bigabs}[1]{\halfspace \Big \vert \halfspace #1 \Big \vert}

\newcommand{\norm}[1]{{\|#1\|}}                             
\newcommand{\bignorm}[1]{\big\| #1 \big\|}

\newcommand{\halfspace}{\kern 0.2em}
\newcommand{\qspace}{\kern 0.1em}

\newcommand{\edit}[1]{\textcolor{black}{#1}}

\usepackage[compress]{cite}

\theoremstyle{plain}
\newtheorem{theorem}{Theorem}
\newtheorem{lemma}{Lemma}

\newtheorem{assumption}{Assumption}
\newtheorem{definition}{Definition}

\title{\LARGE \bf
Shaping Large Population  Agent Behaviors Through\\ Entropy-Regularized Mean-Field Games
}

\author{Yue Guan$^{1}$,  Mi Zhou$^{2}$, Ali Pakniyat$^{3}$, and Panagiotis Tsiotras$^{4}$
\thanks{$^{1}$ Yue Guan is a PhD student with the School of Aerospace Enginnering, Georgia Institute of Technology, Atlanta, GA, USA
{\tt\small yguan44@gatech.edu}
}
\thanks{$^{2}$Mi Zhou is a PhD student with the School of Electrical and Computer Engineering, Georgia Institute of Technology, Atlanta, GA, USA. 
{\tt\small mzhou91@gatech.edu}}
\thanks{$^{3}$Ali Pakniyat is an Assistant Professor with the department of Mechanical Engineering, University of Alabama, Tuscaloosa, AL, USA.
{\tt\small apakniyat@ua.edu}}
\thanks{$^{4}$Panagiotis Tsiotras is the David \& Andrew Lewis Chair Professor with the School of Aerospace Engineering, Georgia Institute of Technology, Atlanta, GA, USA.
{\tt\small tsiotras@gatech.edu}}
}

\begin{document}

\renewcommand{\baselinestretch}{0.98}

\maketitle
\thispagestyle{empty}
\pagestyle{empty}

\begin{abstract}
Mean-field games (MFG) were introduced to efficiently analyze approximate Nash equilibria in large population settings.
In this work, we consider entropy-regularized mean-field games with a finite state-action space in a discrete time setting. 
We show that entropy regularization provides the necessary regularity conditions, that are lacking in the standard finite mean field games.
Such regularity conditions enable us to design fixed-point iteration algorithms to find the unique mean-field equilibrium (MFE).
Furthermore, the reference policy used in the regularization provides an extra parameter, through which one can control the behavior of the population.
We first consider a stochastic game with a large population of $N$ homogeneous agents. 
We establish conditions for the existence of a Nash equilibrium in the limiting case as $N$ tends to infinity, and we demonstrate that the Nash equilibrium for the infinite population case is also an $\epsilon$-Nash equilibrium for the $N$-agent system, where the sub-optimality $\epsilon$ is of order $\mathcal{O}\big(1/\sqrt{N}\big)$.
Finally, we verify the theoretical guarantees through a resource allocation example and demonstrate the efficacy of using a reference policy to control the behavior of a large population. 
\end{abstract}

\section{Introduction}

Decision making in decentralized systems arises in many applications, ranging from
multi-robot task allocation~\cite{berman2009optimized, korsah2013comprehensive,shishika2021dynamic,zhou2021game}, 
finance~\cite{lehalle2019mean, fu2021mean,lachapelle2016efficiency}, etc.
The scalability of the solution to large populations is an important consideration in these settings, as the complexity of the system increases drastically with the number of agents. 

To address the scalability issues, the mean field approach was introduced in~\cite{huang2006large, huang2007large, lasry2007mean}. 
The mean-field game (MFG) formulation reduces the interactions among agents to a game between a \textit{representative agent} and a population of infinitely many other agents.
Such a population is often referred to as the \textit{mean field}, 
and the solution in this limiting case is the mean-field equilibrium (MFE).
In the continuous setting, the MFE is characterized by a Hamilton-Jacobi-Bellman equation (HJB) coupled with a transport equation.
The HJB equation describes the optimality conditions for the policy of the representative agent, and the transport equation captures the evolution of the population distribution. 
Furthermore, the optimal policy computed by the representative agent constitutes an $\epsilon$-Nash equilibrium when all the agents in the \textit{finite} $N$-population deploy this policy, for some sufficiently large $N$.
The existence and uniqueness of such an optimal policy have been established in~\cite{huang2006large}.

Although the discretization of continuous MFG has been studied in prior works~\cite{achdou2010mean}, direct analysis results for discrete-time and finite state-action space MFG are still relatively sparse. 
One of the challenges in the finite MFG is the absence of regularity conditions regarding the mean field~\cite{Guo2019LearningMG}. 
That is, when the population mean field changes slightly, the corresponding optimal policy for the representative agent could change drastically~\cite{cui2021approximately}.
Previous works have used Boltzmann policies~\cite{yang2018mean} and projection to meshed probability measure spaces~\cite{Guo2019LearningMG} to avoid such issues.
More recent works~\cite{anahtarci2020q} directly introduced a relative entropy term to the reward structure to provide regularity conditions.
The authors in~\cite{cui2021approximately} used entropy-regularization to stabilize the iterative algorithm and reduced regularization over time to learn the original MFE.
The existence and uniqueness of the entropy-regularized MFE was also examined.

Different from these previous works, in this paper we explicitly consider the reference policy in the entropy-regularization as an extra feature that allows us to control the behavior of a large population.
Consider the situation, for instance, where a ``coordinator'' of a large population of agents desires to impose a certain group behavior, but it does not have access to the actual rewards. 
The agents are selfish and not concerned about the overall performance of the population, but they have access to the actual rewards. 
If the coordinator designs a policy and forces the whole population to adopt it, the result could be undesirable, as such a policy will not be informed of the actual agent rewards. 
Without a reference policy, however, agents may fail to find the MFE, or they may find an MFE that does not induce a desirable group behavior.
We argue that the entropy-regularized MFG is a good framework to model such scenarios.
Through a resource allocation example, we show that one can encode desired population behaviors into the reference policy.  
By adjusting the multiplier of the regularization term, one can further produce a tunable behavior of the population that balances the encoded behavior and the cumulative rewards.

\textit{Contributions:}
In this work, we formulate a game of $N$-homogeneous agents and show that under pairwise coupled rewards, the state of the system can be \textit{exactly} represented by a distribution over the state space.
We then consider the limiting infinite population game and introduce entropy regularization to construct contraction operators to find the unique regularized MFE. 
We consider a special class of MFGs where the agents have coupled rewards but decoupled dynamics.
For this class of MFGs, we streamline and simplify the convergence proof in~\cite{cui2021approximately}. 
Finally, we verify the theoretical results through a numerical example for a resource allocation problem and demonstrate that certain performance is not possible without entropy-regularization and a properly selected reference policy.



\section{Problem Formulation}

Consider a large population game consisting of $N$ homogeneous agents, where  $N\gg1$.
We define the game through the tuple $\langle \S, \A, \T, \R^1, \ldots, \R^N, N, H \rangle$. 
The game is over a discrete-time with a finite horizon~$H$.
In this formulation, we assume that all agents share the same finite state space $\S$ and the same finite action space $\A$.
At time $t$, agent $i$ takes an action $a_t^i \in \A$ and transitions from state $s^i_t$ to $s^i_{t+1}$ according to the dynamics~$\T$, which we will discuss in greater detail later on.
As a consequence of its own action, agent $i$ receives a reward $\R^i_t(s^i_t,a^i_t,s^{-i}_t)$, where $s^{-i}_t$ is a shorthand notation for $(s^1_t,\ldots,s^{i-1}_t,s^{i+1}_t,\ldots,s^N_t)$.
Each agent follows a (time-varying) Markov policy $\pi^i = \{\pi^i_t\}_{t=0}^{H}$, 
such that, at each time step $t$, 
this policy is a mapping $\pi^i_t: \S \times \A \to [0,1]$.
We use $\Pi$ to denote the space of admissible policies. 
For simplicity, we use $t\leq H$ to denote that $t \in \{0,\ldots, H\}$.


\noindent 
\paragraph{Dynamics}
We assume that all agents have the same decoupled dynamics $\T : \S \times \S\times \A \to [0,1]$.
The value of $\T(s_{t+1}|s_{t},a_t)$ represents the probability of transitioning from state $s_t$ to state $s_{t+1}$ under action $a_t$.\footnote{%
\edit{We use $s^i\in \S$ to denote the state of the specific agent $i$, and $s \in \S$ to denote the state of a generic agent.
Similar rules apply to $a^i$ and $a$.}}
%
In the sequel, we use the notation $\T(\cdot|\cdot,\pi_t^i)$ to denote the dynamics of agent $i$ following the policy $\pi^i_t$.
Formally, 
\begin{equation}\label{eqn:policy-induced-transition}
    \T(s^i_{t+1}|s^i_t,\pi_t^i) = \sum_{a^i_t\in\A} \T(s^i_{t+1}|s^i_t, a^i_t) \halfspace \pi_t^i(a^i_t|s^i_t).
\end{equation}

Assuming that all agents start with the same initial state distribution $\mu_0$ and that each agent $i$ deploys a policy $\pi^i$,
we have $N$ independent processes, where the $i$-th process follows the dynamics
\begin{equation} \label{eqn:n-processes}
         s_{t+1}^i \sim \T\big(\cdot|s^i_{t}, \pi^i_{t}\big), 
         ~~~~ t = 0,\ldots, H-1, \qquad s_0^i \sim \mu_0.
\end{equation}

\paragraph{Rewards}
The reward of agent $i$ at time $t$ is
\begin{equation}\label{eqn:rewards}
    \R^i_t(s^i_t, a^i_t, s^{-i}_t) = \Theta_t \Big(\frac{1}{N} \sum_{k=1}^N L_t(s^i_t, a^i_t, s^k_t) \Big),
\end{equation}
where $\Theta_t:\mathbb{R} \to \mathbb{R}$ and $L_t :\S \times \A \times \S \to \mathbb{R} $.
%
\begin{assumption} \label{assmpt:reward-lip}
    The function $\Theta_t$ is uniformly globally Lipschitz continuous in $t$ with Lipschitz constant $K_\Theta$. 
    That is, for all $x,y \in \mathbb{R}$ and $t \leq H$,
$
        |\Theta_t(x) - \Theta_t(y) | \leq K_\Theta \abs{x-y}.
$
\end{assumption}

\begin{assumption} \label{assmpt:reward-bound}
There exists a positive $L_{\max}$ such that, for all $s,s'\in \S, a\in \A \text{ and } t \leq H$,
$|L_t(s,a,s')| \leq L_{\max}.$
\end{assumption}

The maximum magnitude of the reward is
$\R_{\max} =\max_{|x|\leq L_{\max}, t\leq H} |\Theta_t(x)|$.
Note that the reward in~\eqref{eqn:rewards} is indifferent to the ordering of the agents.
As a consequence, agent $i$'s reward can be computed, given only the fraction of agents at each state.
This observation motivates the \textit{aggregation} of the system state $(s^1_t,\ldots, s^N_t)$ to a distribution of the agents' states.

\paragraph{Empirical distribution}

For the $N$ processes in~\eqref{eqn:n-processes}, we define the empirical distribution at time $t$ as
\begin{equation}\label{eqn:empirical-dist}
    \mu^N_t  (s) = \frac{1}{N} \sum_{k=1}^{N} \mathds{1}_s (s^k_t),
    \quad s \in \S,
\end{equation}
where $\mathds{1}_x$ is the indicator function, i.e., $\mathds{1}_x (y) = 1$ if $y=x$, and $0$ otherwise.
The empirical distribution flow is then defined as $\mu^N = \{\mu^N_t\}_{t=0}^H$.
Note that $\sum_{s}\mu^N_t(s) = 1$, and thus $\mu^N_t$ is a probability measure over $\S$.
We denote the space of probability measures over $\S$ as $\P(\S)$.
Then, $\M = \big(\P(\S)\big)^{H+1}$ is the space of the probability measure flows and $\mu^N \in \M$.

\paragraph{Metric spaces}
We use total variation~\cite{grimmett2001probability} as the metric for the probability measure space $\P(\X)$.
When $\X$ is finite, the total variation between $\nu,\nu' \in \P(\X)$ is given by
$
    \dTV(\nu, \nu') = \frac{1}{2} \sum_{x\in\X} \abs{\nu(x)-\nu'(x)} = \frac{1}{2} \bignorm{\nu(x)-\nu'(x)}_1.
$
We equip both $\M$ and $\Pi$ with the supremum metric induced by the total variation. 
That is, for $\mu, \mu' \in \M$, we define
$
    \dM(\mu,\mu') = \max_{t\leq H} \; \dTV \; \left(\mu_t, \mu'_t\right),
$
and for policies $\pi,\pi' \in \Pi$
$
    \dPi(\pi,\pi') = \max_{t\leq H} \max_{s\in\S} \; \dTV \left(\pi_t (s), \pi'_t(s)\right),
$
where $\pi_t(s) \in \P(\A)$ is the distribution the policy assign over actions when an agent is at state $s$. 
It can be shown that both $(\M,\dM)$ and $(\Pi, \dPi)$ are complete metric spaces~\cite{grimmett2001probability}.

\paragraph{Distribution induced rewards}

Due to symmetry, the state-coupled reward in~\eqref{eqn:rewards} can be characterized through the empirical distribution. 
Overloading the notation, we also write
\begin{align*}
    L_t(s^i_t, a^i_t,\mu_t^N)
    &\overset{\Delta}{=} \sum_{s'\in \S} L_t(s^i_t, a^i_t,s')\mu_t^N(s') \label{eqn:reward-mu}\\
    &= \sum_{s' \in \S} L_t(s^i_t, a^i_t,s') 
    \Big(\frac{1}{N} \sum_{k=1}^{N} \mathds{1}_{s'} (s^k_t)\Big)\\
    &= \frac{1}{N} \sum_{k=1}^{N} \Big(\sum_{s' \in \S} L_t(s^i_t, a^i_t,s') 
     \mathds{1}_{s'} (s^k_t)\Big)\\
    &=\frac{1}{N}\sum_{k=1}^N L_t(s^i_t, a^i_t, s^k_t).
\end{align*}
With the above definition, we further define the distribution-induced reward for each agent as
\begin{equation}\label{eqn:rewards-empirical-dist}
\begin{aligned}
        \R^i_t(s^i_t, a^i_t, \mu^N_t) & \overset{\Delta}{=}
        \Theta_t \Big(L_t(s^i_t, a^i_t,\mu_t^N) \Big).
\end{aligned}
\end{equation}



\begin{lemma}\label{lmm:reward-continuity}
    The reward function $\R^i_t(s,a,\mu)$ in (\ref{eqn:rewards-empirical-dist}) is globally Lipschitz with respect to the probability measure $\mu \in \P(\S)$, with Lipschitz constant $2K_\Theta \halfspace L_{\max}$. 
\end{lemma}

\begin{proof}
One can verify that, for all $\nu,\nu' \in \P(\S)$,
\begin{equation*}
    | L_t(s, a,\nu) - L_t(s, a,\nu') | \leq 2 L_{\max} \dTV\left(\nu,\nu'\right).
\end{equation*}
Then, through composition with Lipschitz function $\Theta_t$, the desired Lipschitz constant for $\R^i_t(s,a,\nu)$ can be shown.
\end{proof}

\paragraph{Expected cumulative reward}

The expected cumulative reward of agent $i$ induced by the joint policies 
$(\pi^1,\ldots, \pi^N)$ is given by 
\begin{equation}\label{eqn:cumulative-reward}
\begin{aligned}
J^{i,N} (\pi^i, \pi^{-i}) 
=  \mathbb{E}\Big[\sum_{t=0}^{H} \R^i_t(s^i_t, a^i_t, \mu^N_t) \Big],
\end{aligned}
\end{equation}
where the expectation is taken over the system trajectories with each agent $i$ starting with initial distribution~$\mu^N_0$ and following the policy $\pi^i$.
Each agent's objective is to select a policy that maximizes its own expected cumulative rewards.
We therefore have $N$ coupled optimization problems:
\begin{equation}
    \max_{\pi^i \in \Pi} ~ J^{i,N} (\pi^i, \pi^{-i}), \qquad  i = 1,\ldots,N.
\end{equation}

One of the most common solution concepts for a game with such coupled optimization is the Nash equilibrium~\cite{owen:Game-Theory}. 

\begin{definition}
A Nash equilibrium (NE) is a tuple $(\pi^{1*},\ldots, \pi^{N*})$ such that, for all $i =1,\ldots, N$, 
\begin{equation*}
J^{i,N} (\pi^{i},\pi^{-i*}) \leq J^{i,N} (\pi^{i*}, \pi^{-i*}),  ~~~ \forall \pi^{i}\in \Pi.
\end{equation*}
\end{definition}
\begin{definition}
For $\epsilon \geq 0$, an $\epsilon$-Nash equilibrium is a tuple $(\pi^{1*},\ldots, \pi^{N*})$ such that, for all $i =1,\ldots, N$, 
\begin{equation} \label{eqn:eps-Nash-def}
 J^{i,N} (\pi^{i},\pi^{-i*})\leq J^{i,N} (\pi^{i*}, \pi^{-i*}) +\epsilon, \quad ~ \forall \pi^{i}\in \Pi.
\end{equation}
\end{definition}
\noindent In other words, any unilateral deviation from an $\epsilon$-NE can improve an agent's performance by at most $\epsilon$. 

In this work, we further restrict our attention to identical policies for all agents. 

\begin{assumption}\label{assmpt:identical-policy}
  For all $i,j\in\{1,\ldots,N\}$,   $\pi^i = \pi^j$.
\end{assumption}
\noindent 
Assumption~\ref{assmpt:identical-policy} leads, in general, to a loss in performance~\cite{arabneydi2014team}. 
However, identical policy is a standard assumption in the literature of large scale systems for reasons of simplicity and robustness~\cite{shi2012survey}.
In light of Assumption~\ref{assmpt:identical-policy}, henceforth, we will drop the superscripts on the policies and denote the policy used by all agents as $\pi$.

\section{Mean Field Approximation}\label{sec:mean-field}
As $N$ approaches infinity, the limiting game constitutes the mean field game. 
The mean field is defined as the empirical distribution of the infinite population.
We denote the mean field at time $t$ as $\mu_t$.
Aside from describing the infinite population, the introduction of the mean field also has attractive computational benefits.
Recall that the empirical distribution in~\eqref{eqn:empirical-dist} is a \textit{random} vector. 
To properly evaluate the expected reward with the nonlinear function $\Theta_t$ in \eqref{eqn:rewards}, one needs the distribution of $\mu_t^N$ at each time step. 
In general, the propagation of the distribution of $\mu_t^N$ is  computationally expensive. 
On the other hand, under the identical policy $\pi$ used by all  (infinite in number) agents, the trajectory of the mean field is deterministic~\cite{anahtarci2020q}.
Furthermore, $\mu_t$ follows a simple propagation rule:
\begin{equation}\label{eqn:mean-field-prop}
    \mu_{t+1} = \mu_t \left[\T(\pi_t)\right],
\end{equation}
where $\left[\T(\pi_t)\right]$ is a right stochastic matrix constructed based on~\eqref{eqn:policy-induced-transition}.
We refer to the time sequence $\mu= \{\mu_t\}_{t=0}^H \in \M$ as the \textit{mean field flow}.

It is tempting to approximate the empirical distribution of a finite $N$-population with the mean field. 
Indeed, as we will show later, the empirical distribution converges to the mean field as the number of agents approaches infinity.

\subsection{Representative Agent}

Before tackling the large population game with $N$ agents, we consider the limiting infinite population case by specifying the behaviour of the representative agent. 
Since the effect of dynamic uncertainties on all agents takes the same form, the mean field flow $\mu$ can be solely generated from the representative agent dynamics and its policy.
Assuming that the mean field flow $\mu$ is known and fixed, this yields
a standard Markov Decision Process (MDP) $\langle \S, \A, \T, \R_{\mu}\rangle$.
The state space, action space and the transitions of the induced MDP come directly from the original game. 
The reward induced by the mean field $\mu$ is given by
\begin{equation}\label{eqn:induced-MDP-reward}
    \R_{\mu,t} (s,a) = \Theta_t\big(L_t(s,a,\mu_t) \big).
\end{equation}
The representative agent can then maximize its expected cumulative reward given the mean field flow $\mu$ as follows: 
\begin{equation}\label{eqn:mf-cumulative-reward}
\begin{aligned}
J_{\mu}(\pi^*) = & \max_{\pi\in\Pi} \; \mathbb{E}_{\mu_0}\Bigg[\sum_{t=0}^{H} \R_{\mu,t}(s_t, a_t) \Bigg].
\end{aligned}
\end{equation}
\noindent 
Note that the optimal policy depends on $\mu$. 
We use the operator $\B_\opt :\M \to \Pi$ to denote the mapping from the mean field flow to an optimal policy of the 
induced MDP:\footnote{In general, $\B_\opt$ is a set-valued function, since the optimal policy of an MDP need not be unique.} 
\begin{equation}\label{eqn:greedy-mdp-opt}
    \pi^{*} = \B_\opt \left(\mu\right).
\end{equation}

When all agents employ the policy $\pi$ of the representative agent, 
a new mean field flow $\mu$ is induced and can be propagated via~\eqref{eqn:mean-field-prop} starting from $\mu_0$. 
We use the operator $\B_\prop : \Pi \to \M$ to denote this propagation. That is,
\begin{equation}
    \mu = \B_\prop \left(\pi \right).
\end{equation}

The mean-field equilibrium (MFE) of a mean-field game is defined as a consistent pair $(\pi^*, \mu^*) \in \Pi \times \M$ such that
\begin{equation}\label{eqn:consistent-condition}
         \pi^* = \B_\opt \big( \mu^* \big), \quad  \mu^* = \B_\prop(\pi^*).
\end{equation}
The existence of such consistent pair can be established through a Brouwer's fixed point argument~\cite{saldi2018markov}. 

One may attempt to use fixed-point iterations to find a solution to~\eqref{eqn:consistent-condition}. 
Unfortunately, the composed mean-field equilibrium (MFE) operator $\Gamma =\B_\prop \circ \B_\opt$ is only non-expansive and not a contraction, in general. 
\edit{The non-contractiveness of the MFE operator $\Gamma$ comes from the \textit{hard} maximization within the MDP optimization $\B_\opt$, where slightly different induced rewards may lead to significantly different optimal policies. 
Consequently, even if two initial mean-field flows are close, they may induce totally different optimal policies, which then leads to two new mean-field flows that are far apart. 
}
For a more detailed discussion, one may refer to~\cite{cui2021approximately}.

\section{Entropy-Regularized Mean Field Games} \label{sec:ER-MFG}

Entropy regularization is a technique used extensively to stabilize learning algorithms and to reduce the maximization bias~\cite{Fox:2015}.
The extra entropy cost introduced to the reward structure prevents abrupt policy changes between iterations. 

To address the issue of non-contractiveness of the operator $\Gamma$, we introduce an entropy-regularization term and replace the hard maximization in $\B_\opt$ with a soft maximization. 
With the soft maximization, a small change in the mean field flow will not result in an abrupt change in the optimal policy of the representative agent, thus inducing a \edit{contractive} MFE operator. 

Given a reference policy  $\rho \in \Pi$ such that $\rho_t(a|s)>0$ for all $t \leq H, s\in \S$, and $a\in\A$,
we introduce an entropy regularization term to~\eqref{eqn:mf-cumulative-reward} as follows:
\begin{align} \label{eqn:soft-cumulative-reward}
J^\KL_{\mu}&(\pi;\rho)
= \mathbb{E}\Bigg[\sum_{t=0}^{H} \Big( \R_{\mu,t}(s_t, a_t)
- \frac{1}{\beta} \log \frac{\pi_t(a_t|s_t)}{\rho_t(a_t|s_t)}\Big) \Bigg], \nonumber
\end{align}
where $\beta >0$ is the inverse temperature, and it is a design parameter.
The reference policy can encode any preference one has about the population behavior.\footnote{
When no information encoding is needed, one can use a uniform prior. }
When $\beta$ is small, the regularization term 
is dominant, and $\pi$ approaches the reference~$\rho$. 
When $\beta$ is large, the agent is allowed to diverge from the reference policy to increase the collected rewards.
As a consequence, the optimal $\pi$ approaches a greedy policy produced by $\B_\opt$ as in~\eqref{eqn:greedy-mdp-opt}.

\subsection{Optimization for the Regularized MDP}

It can be shown that the \textit{unique} optimal policy that solves $\max_{\pi} J^\KL_{\mu} (\pi;\rho)$ is given by the following (weighted) Boltzmann distribution~\cite{Fox:2015}  
\begin{equation}\label{eqn:soft-policy}
    \pi^\KL_{\mu, t}(a|s) = \frac{1}{Z_t (s)} \rho_t(a|s) \exp\left[\beta Q^\KL_{\mu,t} (s,a;\rho)\right],
\end{equation}
where $Z_t(s)$ is a normalization factor, and
the entropy-regularized state-action value function $Q^{\KL}_{\mu,t}$ is computed as:
\begin{align*}
    & Q^\KL_{\mu,t} (s,a; \rho) = \R_{\mu,t}(s,a)     
     \\
    & \resizebox{\hsize}{!}{
    $+ \sum_{s'} \T(s'|s,a)  \Big(\frac{1}{\beta} \log \sum_{a} \rho_t(a|s)\exp \left[\beta Q^\KL_{\mu,t+1} (s,a; \rho)\right]\Big)$,}
\end{align*}
with the boundary condition $Q^\KL_{\mu,H}(s,a;\rho) = \R_{\mu,H} (s,a)$.

As the entropy-regularized optimal policy is induced by the given mean field, we denote the operation performed in~\eqref{eqn:soft-policy} using the operator $\B^\KL_{\opt,\beta} :\M \to \Pi$.
We can then define the entropy-regularized MFE (ER-MFE) operator $\Gamma^{\KL}_\beta : \M \rightarrow \M$
as 
\begin{equation}\label{eqn:ER-MFE-operator}
    \Gamma^{\KL}_\beta = \B_\prop \circ \B^\KL_{\opt,\beta},
\end{equation}
The regularized equilibrium is then defined as follows.
\begin{definition}
The entropy-regularized mean field equilibrium (ER-MFE) is a consistent pair $(\pi^{\KL*}, \mu^{\KL*}) \in \Pi \times \M$ such that  $\pi^{\KL*} = \B^{\KL}_{\opt,\beta}(\mu^{\KL*})$ and $\mu^{\KL*} = \B_{\prop}(\pi^{\KL*})$.
\end{definition}

In the sequel, we establish the existence and uniqueness of the ER-MFE.
The main goal is to show that the ER-MFE operator is a contraction if $\beta$ is selected properly.
The following derivation is more direct and easier to demonstrate than the one reported in~\cite{cui2021approximately}.
The reason is that we are restricting ourselves to the case where the agents have decoupled dynamics. 
As a consequence, a single agent deviation from the optimal policy does not directly impact the distribution of the rest of the population.

\subsection{Convergence Analysis}

We first establish the Lipschitz continuity of the operators $\B_\prop$ and $\B^{\KL}_{\opt,\beta}$. 
\begin{lemma}\label{lmm:B-prop}
For all $\pi, \pi' \in \Pi$, we have that
\begin{equation}
    \dM(\B_\prop(\pi), \B_\prop(\pi')) \leq K_\prop \; \dPi(\pi, \pi'),
\end{equation}
where
\begin{equation}
    K_\prop = \frac{|S|(|S|^H -1)}{|S|-1}.
\end{equation}
\end{lemma}

\begin{proof}
    See the appendix.
\end{proof}

The following two Lemmas are adopted from~\cite{cui2021approximately}. 

\begin{lemma}[\cite{cui2021approximately}]
Under Assumptions~\ref{assmpt:reward-lip} and~\ref{assmpt:reward-bound}, the entropy-regularized Q-function $Q^\KL_{\mu}$ is Lipschitz with respect to $\mu$ for arbitrary $\beta >\beta_{\max}$ and $\beta_{\max}>0$. That is,
\begin{equation*}
    \max_{t,s,a} \bigabs{Q^\KL_{\mu,t} (s,a; \rho)-Q^\KL_{\mu',t} (s,a; \rho)} \leq K^\KL_Q \;\dM\left(\mu,\mu'\right),
\end{equation*}
where
 $K^\KL_Q = \max_{t\leq H} K^\KL_{Q,t}$,
and $K^\KL_{Q,t}$ is defined as 
\begin{equation*}
\resizebox{\hsize}{!}{
    $K^\KL_{Q,t} = 2K_\Theta L_{\max} + \frac{\rho_{\max} \exp\left(2 \beta_{\max} (H+1)\R_{\max} K^\KL_{Q,t+1}\right)}{\rho_{\min}}$,}
\end{equation*}
with boundary condition $K^\KL_{Q,H} = 2K_\Theta L_{\max}$, and  
$\rho_{\max} =\max_{t,s,a}\rho_t(a|s) >0$, $\rho_{\min} = \min_{t,s,a}\rho_t(a|s) >0$.
\end{lemma}


\begin{lemma}[\cite{cui2021approximately}]
Under Assumptions~\ref{assmpt:reward-lip} and~\ref{assmpt:reward-bound}, the entropy-regularized operator $\B^\KL_\opt$ is Lipschitz, that is,
\begin{equation*}
    \dPi\left(\B^\KL_{\opt,\beta}(\mu),\B^\KL_{\opt,\beta}(\mu')\right) \leq K^\KL_{\opt,\beta} \; \dM\left(\mu,\mu'\right),
\end{equation*}
where,
\begin{equation}
    K^\KL_{\opt,\beta} = \frac{|\A|(|\A|-1)  \beta \rho_{\max}^2}{2\rho_{\min}^2} K^\KL_Q.
\end{equation} 
\label{lemma:cui}
\end{lemma}
The Lipschitz continuity in Lemma \ref{lemma:cui} guarantees that a small change in the mean field can only result in a small change in the optimal policy.
With the Lipschitz constants of $\B_\prop$ and $\B_{\opt,\beta}^{\KL}$, we arrive at the following result regarding the selection of $\beta$ to ensure that $\Gamma^{\KL}_\beta$ is a contraction.

\begin{theorem}
The entropy-regularized mean-field equilibrium (ER-MFE) operator $\Gamma_\beta^\KL = \B_\prop \circ \B^\KL_{\opt,\beta}$ is a contraction for 
\begin{equation} \label{eqn:beta:con}
    \beta < \min \Bigg\{\beta_{\max}, \frac{2 \rho_{\min}^2}{\rho_{\max}^2 |\A|(|\A|-1) } \frac{1}{K^\KL_Q K_\prop}\Bigg\}.
\end{equation}
\end{theorem}

\begin{proof}
The $\beta$ selection in~(\ref{eqn:beta:con})
guarantees that ${K_\prop  K^\KL_{\opt,\beta} < 1}$.
Consequently, the ER-MFE operator $\Gamma_\beta^\KL$ is a contraction.
\end{proof}

\subsection{Error Bounds on the Mean Field Approximations}\label{sec:bounds}
\edit{With the theoretical results established for the infinite-population equilibrium, we now examine the performance guarantee for applying the infinite-population equilibrium to a finite $N$-agent system.}
We first present the following lemma, which characterizes the asymptotic convergence of the empirical distribution flow $\mu^N$ to the mean field flow $\mu$. 

\begin{lemma}\label{lmm:mean-field-bound}
Suppose a mean field flow $\mu$ (infinite population)  is induced by the policy $\pi$.
Let $\mu^N$ denote the empirical distribution of an $N$-agent system, where all agents deploy the same policy $\pi$. 
Then, for all $t \leq H$,
\begin{equation}\label{eqn:mean-field-bound}
    \mathbb{E} \left[\dTV\left(\mu^N_t ,\mu_t\right)\right] = \mathcal{O}\Big(\frac{1}{\sqrt{N}}\Big).
\end{equation}
\end{lemma}

\begin{proof}
See the appendix.
\end{proof}

Next, we show that the ER-MFE $\pi^*$ for the infinite population is an $\epsilon$-Nash equilibrium for the $N$-agent system. 

\begin{theorem}\label{thm:epsilon-nash}
Consider an ER-MFE $(\pi^{*}, \mu^{*})$.\footnote{We drop the superscript $\KL$ for notation simplicity. Unless specified otherwise, $\pi^*$ and $\mu^*$ refer to the entropy-regularized optimal solutions.}
Then, for all $\widetilde{\pi} \in \Pi$, we have
\begin{equation}
    J_{\KL}^{i,N}(\widetilde{\pi},\pi^*) \leq J_\KL^{i,N}(\pi^*, \pi^*) + \mathcal{O}\left( \frac{1}{\sqrt{N}}\right),
\end{equation}
where $J^{i,N}(\widetilde{\pi},\pi^*)$ is the value induced when agent $i$ deviates and applies policy $\widetilde{\pi}$, and all other agents apply policy $\pi^*$.
\end{theorem}

Note that the $N$-agent trajectory under the optimal policy~$\pi^*$ is given by
$s^i_{t+1} \sim \T\big(\cdot| s^i_t, \pi^*_t(s^i_t)\big)$ for $i = 1,\ldots, N$.
Without loss of generality, we let agent 1 deviate and select some other policy $\widetilde{\pi}$. 
For both the optimal system and the deviated system, the agents' initial distributions are $\mu_0$.
Then, the trajectory of the deviated $N$-agent system is given by
\begin{equation}
\label{eqn:deviated-mf-dynamics}
\begin{aligned}
    \widetilde{s}^1_{t+1} &\sim \T\big(\cdot| \widetilde{s}^1_t,\widetilde{\pi}_t(\widetilde{s}^1_t)\big),\\
    \widetilde{s}^i_{t+1} &\sim \T\big(\cdot| \widetilde{s}^i_t, \pi^*_t(\widetilde{s}^i_t)\big), \text{ for } i = 2,\ldots, N.
\end{aligned}    
\end{equation}

The following lemma characterizes the convergence of the deviated empirical distribution to the optimal mean field.
\begin{lemma}\label{lmm:deviated-dist-conv}
    Let $\widetilde{\mu}^{N}$ denote the empirical distribution flow resulted from the system in~\eqref{eqn:deviated-mf-dynamics}.
    Then, for all $t \leq H$,
    \begin{equation*}
        \mathbb{E} \; \Big[\dTV\left(\tilde{\mu}^N_t,\mu^*_t\right)\Big] =\mathcal{O}\Big(\frac{1}{\sqrt{N}}\Big)
    \end{equation*}
\end{lemma}
\begin{proof}
See the appendix.
\end{proof}

We now present a proof for Theorem~\ref{thm:epsilon-nash} using cross disturbance analysis similar to~\cite{huang2006large}.

\begin{proof}[Proof of Theorem~\ref{thm:epsilon-nash}]
%

 For the $N$-agent system, the value of agent 1 induced by its policy deviation is bounded as
 \begin{align}
     &J^{1,N}_\KL(\widetilde{\pi}, \pi^*) 
     = \mathbb{E} \sum_{t=0}^H \left[\R_t\big(\widetilde{s}^1_t,\widetilde{\pi}_t,\widetilde{\mu}_t^N\big)- \frac{1}{\beta} \log \frac{\widetilde{\pi}_t(a_t|\widetilde{s}^1_t)}{\rho_t(a_t|\widetilde{s}^1_t)}\right] \nonumber
     \\
     &\leq \mathbb{E} \sum_{t=0}^H \left[\R_t\big(\widetilde{s}^1_t,\widetilde{\pi}_t,\mu^*_t\big)- \frac{1}{\beta} \log \frac{\widetilde{\pi}_t(a_t|\widetilde{s}^1_t)}{\rho_t(a_t|\widetilde{s}^1_t)} \right]+ \mathcal{O}\Big(\frac{1}{\sqrt{N}}\Big) \label{eqn:eps-nash-2}
     \\
     & \leq \mathbb{E} \sum_{t=0}^H \left[\R_t\big(s^1_t,\pi^*_t,\mu^*_t\big) - \frac{1}{\beta} \log \frac{\pi^*_t(a_t|s^1_t)}{\rho_t(a_t|s^1_t)} \right]+ \mathcal{O}\Big(\frac{1}{\sqrt{N}}\Big) \label{eqn:eps-nash-3} \\
     & = J_\KL(\pi^*,\pi^*) + \mathcal{O}\Big(\frac{1}{\sqrt{N}}\Big). \label{eqn:eps-nash-4}
 \end{align}
 
 
 \noindent In~\eqref{eqn:eps-nash-2}, we used the convergence result in Lemma~\ref{lmm:deviated-dist-conv} and the Lipschitz continuity of $\R_t$ to replace $\widetilde{\mu}_t^N$ with $\mu^*_t$. 
 To arrive at~\eqref{eqn:eps-nash-3}, we used the optimality of $\pi^*$ in the regularized MDP induced by $\mu^*$.
 
 We have shown that the difference between the value of the deviated $N$-agent system and the optimal value of the \textit{infinite population system} is bounded by $\mathcal{O}(1/\sqrt{N})$.
Next, we show that the value of the finite $N$-agent system under the identical optimal policy $\pi^*$ for all agents is also  $\mathcal{O}(1/\sqrt{N})$-close to the optimal value of the infinite population system. 
To see this, note that
\begin{align*}
     &J^{1,N}_\KL(\pi^*, \pi^*) 
     = \mathbb{E} \sum_{t=0}^H \left[\R_t\big(s^1_t,\pi^*_t,\mu_t^N\big)- \frac{1}{\beta} \log \frac{\pi^*_t(a_t|s^1_t)}{\rho_t(a_t|s^1_t)}\right]
     \\
     &\leq \mathbb{E} \sum_{t=0}^H \left[\R_t\big(s^1_t,\pi^*_t,\mu_t^*\big)- \frac{1}{\beta} \log \frac{\pi^*_t(a_t|s^1_t)}{\rho_t(a_t|s^1_t)}\right] + \mathcal{O}\Big(\frac{1}{\sqrt{N}}\Big)
     \\
     & = J_\KL(\pi^*,\pi^*) + \mathcal{O}\Big(\frac{1}{\sqrt{N}}\Big), \label{eqn:eps-nash-4}
 \end{align*}
where the inequality is a result of the Lipschitz continuity of $\R_t$ and the convergence result in Lemma~\ref{lmm:mean-field-bound}. 
One can also lower bound $J^{1,N}_\KL(\pi^*, \pi^*)$ to obtain
  \begin{equation}\label{eqn:dist-optimal-finite-to-infinite}
      \bigabs{J_\KL^{1,N}(\pi^*, \pi^*)- J_\KL(\pi^*,\pi^*)} = \mathcal{O}\big( \frac{1}{\sqrt{N}}\big).
 \end{equation}
 
\noindent Combining~\eqref{eqn:eps-nash-4} and~\eqref{eqn:dist-optimal-finite-to-infinite}, it follows that 
 \begin{equation*}
      J_\KL^{1,N}(\widetilde{\pi}, \pi^*) \leq J_\KL^{1,N}(\pi^*, \pi^*) + \mathcal{O}\big( \frac{1}{\sqrt{N}}\big).
 \end{equation*}
Since the population is homogeneous, the same result applies to all agents,
 thus completing the proof.
\end{proof}

\section{Numerical Example}

In this section, a resource allocation problem is formulated as a mean field game to verify the previous theoretical results. 
Consider a dynamic resource allocation problem~\cite{shishika2021dynamic} over a directed graph $\langle\mathcal{S},\mathcal{E} \rangle$, shown in Fig.~\ref{fig:resource-allocation-graph}.
We use $\S$ and $\mathcal{E}$ to denote the set of nodes and edges, respectively.
A large group of agents traverse through the graph and collect the rewards assigned at the terminal time step $H$, and no running reward is assigned.
At time $H$, if an agent is at state (node)~$3$, then it receives a reward of $1.5$.
If it is at state~$4$, it receives a reward of $1$.
Otherwise, the agent receives no reward. 
At the same time, the agents are penalized for staying at a node with a high population density.
In summary, we have the following state-coupled rewards
\begin{align*}
        L_{H}(s^i, a^i, s^k) &=  \underbrace{1.5 \mathds{1}_{s_3} (s^i) + \mathds{1}_{s_4} (s^i)}_{\text{rewards at states $3$ and $4$}} ~~-
        \underbrace{\mathds{1}_{s^i} (s^k)}_{\substack{\text{penalty of sharing }\\ \text{node with agent $k$}}},\\
        L_t(s^i, a^i, s^k) &= 0, \qquad \text{ for all } t=0,\ldots, H-1.
\end{align*}
We set the nonlinear function in~\eqref{eqn:rewards} to $\Theta_{H}(x) = x^2$. 
Each agent at a state $s$ (graph node)  can choose one of the adjacent states  $s'$ to visit at the next time step.
That is, the action space $\A(s)$ at state $s$ is all the states $s'$ such that $(s,s')\in \mathcal{E}$.

\begin{figure}[h]
\vspace{-5pt}
    \centering
    \includegraphics[width=0.25\textwidth]{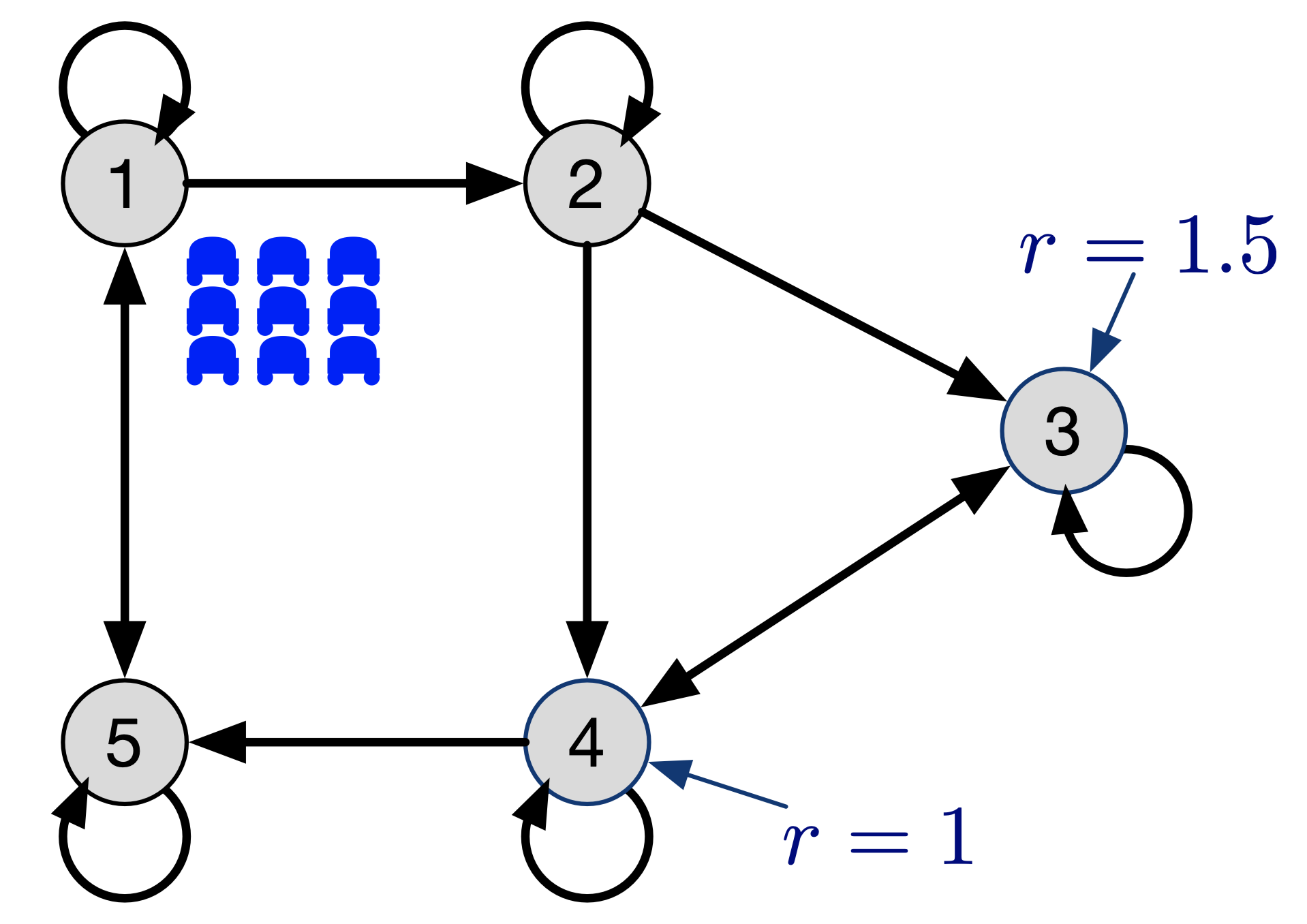}
    \vspace{-8pt}
    \caption{Graph for a resource allocation problem with self-loops.}
    \label{fig:resource-allocation-graph}
    \vspace{-3pt}
\end{figure}

If we directly use fixed-point iterations without entropy-regularization to solve this mean field game, the algorithm {fails to converge}. 
For the first iteration, the agents use a policy that concentrates the whole population at node $3$ for the extra reward.
At the next iteration, the penalty for staying at node $3$ is high given the mean field flow from the previous iteration. 
The representative agent then constructs a policy to visit node $4$.
In summary, the policy found by the $\B_\opt$ without regularization oscillates between concentrating on node $3$ and concentrating on node $4$ at $H$.

Suppose now that a coordinator can send a command to the group of agents, and it decides that both nodes~$3$ and~$4$ need to be occupied by some agents, but s/he does not have access to the actual rewards at the two nodes.
Consequently, the coordinator can, at most, provide a reference policy to \textit{guide} the agents to node $3$ and $4$, but the decision of which node is more rewarding to occupy can only be made by the agents themselves.
We constructed a reference policy $\rho$ that commands the agents to move to nodes $3$ and $4$.\footnote{%
We may set, for example, $\rho(s_4|s_2) = \rho(s_3|s_2) = 0.5$ and $\rho(s_5|s_4)=0.01$.
This reference policy promotes the agents to move from state $2$ to states $3$ and $4$, while discourages agents to move from state $4$ to state $5$.}
If the reference policy is directly applied by the agents, then the final distribution at nodes $3$ and $4$ are roughly the same, which does not reflect the difference in the rewards.

We now use the constructed reference policy to form the entropy-regularized MFG and solve the regularized game using the operator $\Gamma_\beta^\KL$ in~\eqref{eqn:ER-MFE-operator} for two different $\beta$ values.
The algorithm converges and the population distribution over the nodes is depicted in Fig.~\ref{fig:different-beta}.
Recall that a larger $\beta$ means less regularization in the reward structure. 
In Fig.~\ref{fig:different-beta}(a), with a large $\beta$, the agents chase mainly the rewards, and the reference policy from the coordinator has little effect. 
In this scenario, the agents concentrate at node $3$ up to the point when an additional number of agents at the same node would result in a penalty that diminishes the reward advantage that node $3$ has over node $4$.
In Fig.~\ref{fig:different-beta}(b), the value of $\beta$ is small and the reference policy dominates. 
The agents start to ignore the reward advantage that node $3$ has, and follow the reference policy instead.
The parameter $\beta$ enables us to generalize the behavior beyond these two extremes and to cover a continuous spectrum of population behavior. 

\begin{figure}
\vspace{+2pt}
    \centering
    \includegraphics[width=0.99\linewidth]{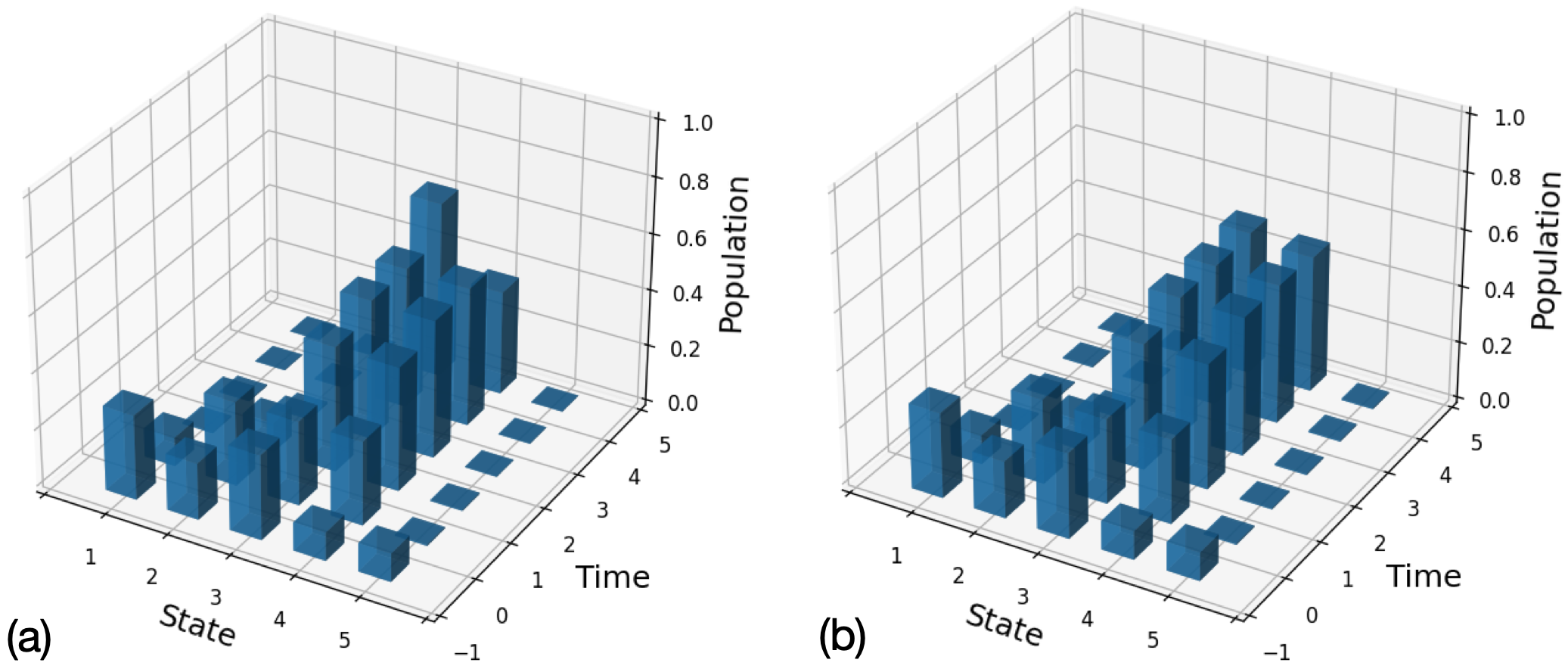}
    \vspace{-10pt}
    \caption{Population distribution over time. Scenario in (a) has the inverse temperature of $\beta=3$, and scenario in (b) has $\beta=0.1$.}
    \label{fig:different-beta}
\end{figure}

Finally, to verify Theorem~\ref{thm:epsilon-nash},
we fixed the last $N-1$ agents' policy to the ER-MFE, and we computed the distribution of the random vector $\mu^N$.
We let the first agent optimize the  entropy-regularized MDP induced by $\mu^N$.
We then compared the difference between its newly-optimized performance and the performance should the agent adopt the ER-MFE.
A log-log plot of performance gain vs. number of agents is presented in Fig.~\ref{fig:deviate}. 
The performance gain trend is bounded by the reference line with a slope of $-0.5$, which verifies our claim of the $\mathcal{O} (1/\sqrt{N})$ convergence rate.

\begin{figure}[h]
    \centering
    \vspace{-3pt}
    \includegraphics[width=0.75\linewidth]{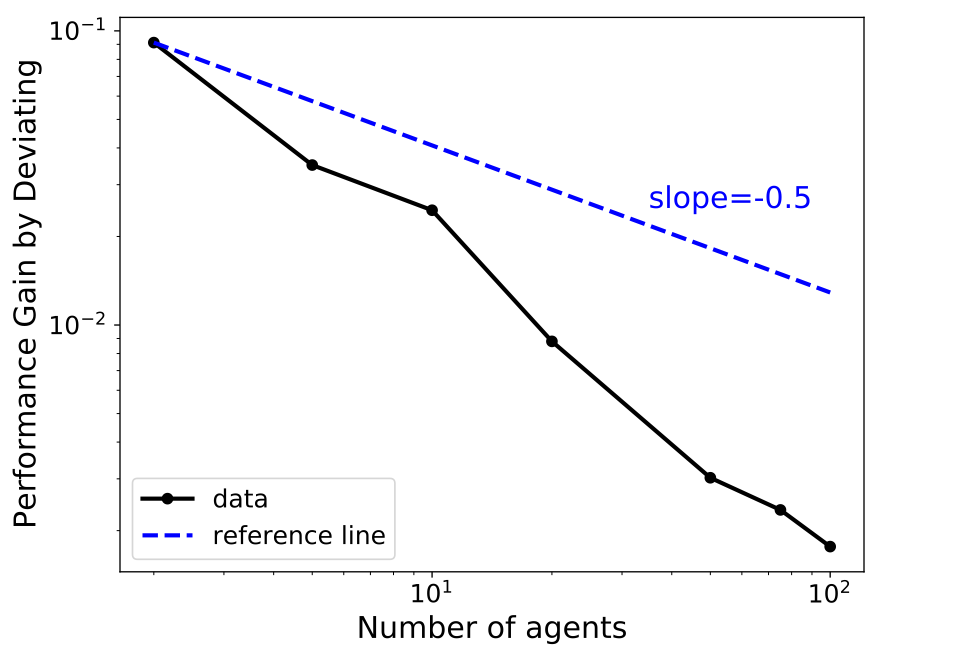}
    \vspace{-8pt}
    \caption{Log-log plot of performance gain by an agent unilaterally deviating in a finite population.}
    \label{fig:deviate}
    \vspace{-10pt}
\end{figure}

\section{Conclusion}
In this article, an entropy-regularized mean-field game with finite state-action space in a discrete time setting was formulated and analyzed.
We demonstrated that the entropy-regularization provides the regularity conditions that the standard MFG lacks. 
The conditions for a \edit{contractive} entropy-regularized mean-field equilibrium operator is presented.
Furthermore, we provided a streamlined proof for the  performance bound of the entropy-regularized MFE in a finite $N$-agent game.  
Through a resource allocation example, we verified the theoretical results and demonstrated that the reference policy for the entropy regularization can be used to control the behavior of a large population, and the parameter $\beta$ allows us to cover a continuous spectrum of population behaviors.
Future work will involve extending the approach to the case of two large teams of agents competing against each other, modeled as a zero-sum game, while the dynamics of agents within each team evolves as a mean field game.
\vspace{+25pt}


\appendix

\begin{proof}[Proof of Lemma~\ref{lmm:B-prop}]
Consider two policies $\pi,\pi' \in \Pi$ and the corresponding propagated mean field flows $\mu = \B_\prop(\pi)$ and $\mu' = \B_\prop(\pi')$.
Then, at time step $t+1$, we have
\begin{align*}
    &\bigabs{\mu_{t+1}(s') - \mu'_{t+1}(s')} \\
    & = \Big \vert \sum_{s}\mu_t(s)\T(s'|s,\pi_t)  -\sum_{s}\mu_t'(s)\T(s'|s,\pi'_t)\Big \vert\\
    &\leq \Big \vert \sum_{s}\mu_t(s)\T(s'|s,\pi_t)  -\sum_{s}\mu_t(s)\T(s'|s,\pi'_t)\Big \vert \tag{A} \label{eqn:A} \\
    &+ \Big \vert \sum_{s}\mu_t(s)\T(s'|s,\pi'_t)  -\sum_{s}\mu_t'(s)\T(s'|s,\pi'_t)\Big \vert. \tag{B}\label{eqn:B}
\end{align*}
For~\eqref{eqn:A}, we have
\begin{align*}
    \text{(A)} &= \Big \vert \sum_{s}\mu_t(s)\sum_{a}\T(s'|s,a) \big(\pi_t(a|s) - \pi_t'(a|s)\big)\Big \vert \\
    &\leq \sum_{s}\mu_t(s)\sum_{a}\T(s'|s,a) \Big \vert  \big(\pi_t(a|s) - \pi_t'(a|s)\big)\Big \vert \\
    &\leq \sum_{s}\mu_t(s)\sum_{a} \Big \vert  \big(\pi_t(a|s) - \pi_t'(a|s)\big)\Big \vert \\
    &= 2 \sum_{s}\mu_t(s)  \dTV(\pi_t(s), \pi'_t(s)) \leq 2 \dPi(\pi,\pi').
\end{align*}

For~\eqref{eqn:B}, we have
\begin{align*}
    \text{(B)} &\leq \sum_{s}\T(s'|s,\pi'_t) \Big \vert   \mu_t(s) -\mu_t'(s)\Big \vert\\
    & \leq \sum_{s} \Big \vert   \mu_t(s) -\mu_t'(s)\Big \vert = 2 \dTV(\mu_t,\mu'_t).
\end{align*}

Combining~\eqref{eqn:A} and~\eqref{eqn:B}, we have
\begin{align*}
    \dTV(\mu_{t+1}, \mu_{t+1}') &= \frac{1}{2} \sum_{s'\in \S} | \mu_{t+1}(s') - \mu'_{t+1}(s') | \\
    &\leq |\S| \big( \dPi(\pi,\pi') + \dTV(\mu_{t}, \mu_{t}') \big).
\end{align*}
For time step $t=0$, we assumed that $\mu_0 = \mu'_0$. 
Consequently, $\dTV(\mu_{0}, \mu_{0}')=0$. 
Through induction, one can show that 
\begin{equation}\label{eqn:appdx-C}
    \dTV(\mu_{t}, \mu_{t}') \leq \frac{|\S|(|\S|^t-1)}{|\S|-1}\dPi(\pi,\pi'),
\end{equation}
Since $|\S| > 1$, it follows that \eqref{eqn:appdx-C} is an increasing sequence of $t$.
Consequently, we have
\begin{align*}
    \dM(\B_\prop(\pi),&\B_\prop(\pi')) = \dM(\mu, \mu')
    = \max_{t\in T} \dTV(\mu_t, \mu_t')\\
    &\leq \frac{|\S|(|\S|^T-1)}{|\S|-1}\dPi(\pi,\pi').
\end{align*}
\end{proof}
\begin{proof}[Proof of Lemma~\ref{lmm:mean-field-bound}]
The empirical distribution is defined as
\begin{equation}
    \vspace{-3pt}
    \mu_t^N (s) = \frac{1}{N}\sum_{i=1}^N \mathds{1}_{s} (s^i_t).
    \vspace{-1pt}
\end{equation}
Let $X_s^i = \mathds{1}_{s} (s^i_t$). 
Since the dynamics are decoupled and all agents use the same policy, with the same initial distribution, the random variables $X_s^i$ are i.i.d. with mean~$\mathbb{E} \left[X_s^i \right]$.
As $\mu^N_t(s)$ is the sample mean of $X^i_s$, from the strong law of large numbers~\cite{grimmett2001probability}, we have
$
    \mu^N_t(s) \xrightarrow{a.s.}    \mathbb{E}\left[X_s^i\right] \text{ as } N \to \infty.
$
Thus, we have that the mean field satisfies
$
    \pr{\mu_t(s) - \mathbb{E}\big[X_s^i\big] \neq 0} = 0.
$
The variance of $X_s^i$ is then
$
\mathrm{Var}(X_s^i)
= \mathbb{E}\left[X_s^i\right] -\big(\mathbb{E}\left[X_s^i\right]\big)^2
=\mu_t(s) \left(1 -\mu_t(s) \right).
$
Here, we regarded $\mu_t(s)$ as a deterministic number and use the property $\mathbb{E}\big[(X_s^i)^2\big] = \mathbb{E} \big[ X_s^i\big]$ as $X_s^i$ is an indicator function. 
Furthermore, 
\begin{align*}
  \mathbb{E}\bignorm{\mu^N_t -\mu_t&}_2^2 = \mathbb{E} \sum_{s\in\S} \abs{\mu^N_t(s) -\mu_t(s)}^2 \\
  & =\mathbb{E} \sum_{s\in\S} \bigabs{\frac{1}{N} \sum_{i=1}^N \left(X_s^i - \mathbb{E}\left[X_s^i\right]\right)}^2 
  \\
  &= \frac{1}{N} \sum_{s\in\S} \mathrm{Var}(X_s^i)
  =\frac{1}{N}\sum_{s\in\S} \mu_t(s)(1-\mu_t(s)) \\
  &= \frac{1}{N} (1- \norm{\mu_t}_2^2)
  \leq \frac{1}{N}.
\end{align*}
By Jensen's inequality,  we have
$
    \mathbb{E}\bignorm{\mu^N_t -\mu_t}_2 \leq {1}/{\sqrt{N}} .
$
Since
$
    \norm{\mu^N_t- \mu_t}_1 \leq \sqrt{|\S|}\; \norm{\mu^N_t- \mu_t}_2,
$
it follows that
\begin{align*}
 \mathbb{E} [\dTV ( &\mu^N_t -\mu_t)] = \mathbb{E} \Big[\frac{1}{2} \sum_s |\mu_t^N(s) - \mu_t^*(s) | \Big]\\ 
 &= \frac{1}{2}\mathbb{E}\norm{\mu - \mu'}_1
 =\mathcal{O}\Big(\frac{1}{\sqrt{N}}\Big).
\end{align*}
\vspace{-5pt}
\end{proof}

\begin{proof}[Proof of Lemma~\ref{lmm:deviated-dist-conv}]
The deviated empirical distribution is given by
  \begin{align}
  \vspace{-5pt}
      \widetilde{\mu}_t^N(s) &= \frac{1}{N} \sum_{k=1}^N \mathds{1}_s (\tilde{s}^k_t) 
      = \frac{1}{N} \mathds{1}_s (\tilde{s}^1_t) + \frac{1}{N} \sum_{k=2}^N \mathds{1}_s (\tilde{s}^k_t). \nonumber
  \end{align}
  Due to the decoupled dynamics, $\tilde{s}^k_t$ follows the same distribution as $s^k_t$ for $k=2,\ldots, N$. 
 Then, the expected difference between the deviated empirical distribution and the original optimal mean field can be bounded as 
 \begin{align}
     \mathbb{E} &\bigabs{\tilde{\mu}^N_t(s) - \mu^*_t(s)} \nonumber\\
     &\leq \mathbb{E} \bigabs{\frac{1}{N} \mathds{1}_s (\tilde{s}^1_t)} + \mathbb{E} \bigabs{\frac{1}{N} \sum_{k=2}^N \mathds{1}_s (\tilde{s}^k_t) - \mu^*_t(s)} \nonumber
     \\
     &\leq \frac{1}{N} + \mathbb{E} \bigabs{\frac{1}{N-1} \sum_{k=2}^N \mathds{1}_s (\tilde{s}^k_t) - \mu^*_t(s)} \label{eqn:deviated-dist-2}\\
     &\qquad \qquad \qquad \qquad \qquad + \mathbb{E} \bigabs{\frac{1}{N(N-1)} \sum_{k=2}^N \mathds{1}_s (\tilde{s}^k_t)} \nonumber\\
     &\leq \frac{1}{N} + \mathcal{O}\Big(\frac{1}{\sqrt{N}}\Big) + \frac{1}{N} 
     = \mathcal{O}\Big(\frac{1}{\sqrt{N}}\Big) \label{eqn:deviated-dist-1}. 
 \end{align}
 The second term in~\eqref{eqn:deviated-dist-2} corresponds to the scenario of $N-1$ agents all applying the optimal policy $\pi^*$. 
 By Lemma~\ref{lmm:mean-field-bound}, we obtain the convergence rate. 
 %
Finally, from~\eqref{eqn:deviated-dist-1}, we have
\begin{align*} 
     &\mathbb{E} \; [\dTV\left(\tilde{\mu}^N_t,\mu^*_t\right)] = \frac{1}{2}\mathbb{E} \; 
     \sum_{s\in\S} |\tilde{\mu}_t^N(s)-\mu_t^*(s) |
     =\mathcal{O}\Big(\frac{1}{\sqrt{N}}\Big).
 \end{align*}
\end{proof}

\bibliographystyle{ieeetr}
\bibliography{references}

\begin{thebibliography}{10}

\bibitem{berman2009optimized}
S.~Berman, A.~Hal{\'a}sz, M.~A. Hsieh, and V.~Kumar, ``Optimized stochastic
  policies for task allocation in swarms of robots,'' {\em IEEE Transactions on
  Robotics}, vol.~25, no.~4, pp.~927--937, 2009.

\bibitem{korsah2013comprehensive}
G.~A. Korsah, A.~Stentz, and M.~B. Dias, ``A comprehensive taxonomy for
  multi-robot task allocation,'' {\em The International Journal of Robotics
  Research}, vol.~32, no.~12, pp.~1495--1512, 2013.

\bibitem{shishika2021dynamic}
D.~Shishika, Y.~Guan, M.~Dorothy, and V.~Kumar, ``Dynamic defender-attacker
  blotto game,'' {\em arXiv preprint arXiv:2112.09890}, 2021.

\bibitem{zhou2021game}
M.~Zhou, Y.~Guan, M.~Hayajneh, K.~Niu, and C.~Abdallah, ``Game theory and
  machine learning in uavs-assisted wireless communication networks: A
  survey,'' {\em arXiv preprint arXiv:2108.03495}, 2021.

\bibitem{lehalle2019mean}
C.-A. Lehalle and C.~Mouzouni, ``A mean field game of portfolio trading and its
  consequences on perceived correlations,'' {\em arXiv preprint
  arXiv:1902.09606}, 2019.

\bibitem{fu2021mean}
G.~Fu, P.~Graewe, U.~Horst, and A.~Popier, ``A mean field game of optimal
  portfolio liquidation,'' {\em Mathematics of Operations Research}, 2021.

\bibitem{lachapelle2016efficiency}
A.~Lachapelle, J.-M. Lasry, C.-A. Lehalle, and P.-L. Lions, ``Efficiency of the
  price formation process in presence of high frequency participants: a mean
  field game analysis,'' {\em Mathematics and Financial Economics}, vol.~10,
  no.~3, pp.~223--262, 2016.

\bibitem{huang2006large}
M.~Huang, R.~Malham{\'e}, and P.~Caines, ``Large population stochastic dynamic
  games: closed-loop {McKean-Vlasov} systems and the {Nash} certainty
  equivalence principle,'' {\em Communications in Information \& Systems},
  vol.~6, no.~3, pp.~221--252, 2006.

\bibitem{huang2007large}
M.~Huang, P.~E. Caines, and R.~P. Malham{\'e}, ``Large-population cost-coupled
  {LQG} problems with nonuniform agents: individual-mass behavior and
  decentralized $\epsilon$-{Nash} equilibria,'' {\em IEEE Transactions on
  Automatic Control}, vol.~52, no.~9, pp.~1560--1571, 2007.

\bibitem{lasry2007mean}
J.-M. Lasry and P.-L. Lions, ``Mean field games,'' {\em Japanese Journal of
  Mathematics}, vol.~2, no.~1, pp.~229--260, 2007.

\bibitem{achdou2010mean}
Y.~Achdou and I.~Capuzzo-Dolcetta, ``Mean field games: numerical methods,''
  {\em SIAM Journal on Numerical Analysis}, vol.~48, no.~3, pp.~1136--1162,
  2010.

\bibitem{Guo2019LearningMG}
X.~Guo, A.~Hu, R.~Xu, and J.~Zhang, ``Learning mean-field games,'' in {\em
  NeurIPS}, 2019.

\bibitem{cui2021approximately}
K.~Cui and H.~Koeppl, ``Approximately solving mean field games via
  entropy-regularized deep reinforcement learning,'' in {\em International
  Conference on Artificial Intelligence and Statistics}, pp.~1909--1917, PMLR,
  2021.

\bibitem{yang2018mean}
Y.~Yang, R.~Luo, M.~Li, M.~Zhou, W.~Zhang, and J.~Wang, ``Mean field
  multi-agent reinforcement learning,'' in {\em International Conference on
  Machine Learning}, pp.~5571--5580, PMLR, 2018.

\bibitem{anahtarci2020q}
B.~Anahtarci, C.~D. Kariksiz, and N.~Saldi, ``Q-learning in regularized
  mean-field games,'' {\em arXiv preprint arXiv:2003.12151}, 2020.

\bibitem{grimmett2001probability}
G.~Grimmett and D.~Stirzaker, {\em Probability and random processes}.
\newblock Oxford university press, 2001.

\bibitem{owen:Game-Theory}
G.~Owen, {\em Game Theory}.
\newblock Academic Press, 1982.

\bibitem{arabneydi2014team}
J.~Arabneydi and A.~Mahajan, ``Team optimal control of coupled subsystems with
  mean-field sharing,'' in {\em 53rd IEEE Conference on Decision and Control},
  pp.~1669--1674, IEEE, 2014.

\bibitem{shi2012survey}
Z.~Shi, J.~Tu, Q.~Zhang, L.~Liu, and J.~Wei, ``A survey of swarm robotics
  system,'' in {\em International Conference in Swarm Intelligence},
  pp.~564--572, Springer, 2012.

\bibitem{saldi2018markov}
N.~Saldi, T.~Basar, and M.~Raginsky, ``{Markov--Nash} equilibria in mean-field
  games with discounted cost,'' {\em SIAM Journal on Control and Optimization},
  vol.~56, no.~6, pp.~4256--4287, 2018.

\bibitem{Fox:2015}
R.~Fox, A.~Pakman, and N.~Tishby, ``Taming the noise in reinforcement learning
  via soft updates,'' in {\em Proceedings of the 32nd Conference on Uncertainty
  in Artificial Intelligence}, pp.~202--211, 2016.

\end{thebibliography}

\end{document}